\documentclass[11pt, reqno, draft]{amsart}
\usepackage{amssymb,amsbsy,amsthm,amsmath,amstext,amsopn,dsfont}
\usepackage[centering,margin=1.2in]{geometry}

\global\binoppenalty=10000

\newtheorem{theorem}{Theorem}[section]
\newtheorem{lemma}[theorem]{Lemma}
\newtheorem{prop}[theorem]{Proposition}
\newtheorem{cor}[theorem]{Corollary}

\theoremstyle{definition}
\newtheorem{defn}[theorem]{Definition}
\newtheorem{remark}[theorem]{Remark}

\numberwithin{equation}{section}

\newcommand*{\bw}[1]{#1\nobreak\discretionary{}{\hbox{$\mathsurround=0pt #1$}}{}}

\def\beq{\begin{equation}}
\def\eeq{\end{equation}}
\def\beqs{\begin{equation*}}
\def\eeqs{\end{equation*}}

\newcommand{\hr}[1]{\left(#1\right)}                                                    
\newcommand{\ha}[1]{\left\langle#1\right\rangle}                                        
\newcommand{\hs}[1]{\left[#1\right]}                                                    
\newcommand{\hc}[1]{\left\{#1\right\}}                                                  
\newcommand{\br}[1]{\bigl(#1\bigr)}                                                     
\newcommand{\Bc}[1]{\Bigl\{#1\Bigr\}}                                                   

\newcommand{\sco}{,\ldots,}                                                             
\newcommand{\spl}{\bw+\ldots\bw+}                                                       
\newcommand{\sd}{\bw\cdot\ldots\bw\cdot}                                                
\newcommand{\st}{\bw\times\ldots\bw\times}                                              

\newcommand{\suml}[2]{\sum\limits_{{#1}}^{{#2}}}                                        
\newcommand{\sums}[1]{\sum\limits_{{#1}}}                                               

\renewcommand{\ge}{\geqslant}
\renewcommand{\le}{\leqslant}

\newcommand{\subs}{\subset}
\newcommand{\sups}{\supset}
\newcommand{\subseq}{\subseteq}

\newcommand{\longra}{\longrightarrow}

\newcommand{\la}{\leftarrow}

\newcommand{\hra}{\hookrightarrow}                                                      

\newcommand{\eqdef}{\stackrel{\mbox{\tiny\rm def}}{=}}                                  

\def\Ad{\operatorname{Ad}}
\def\a{\mathfrak a}
\def\ad{\operatorname{ad}}
\def\al{\alpha}

\def\C{\mathbb C}

\def\de{\delta}

\def\G{\mathcal G}

\def\Gr{\mathrm{Gr}}
\def\g{\mathfrak g}
\def\ga{\gamma}
\def\gl{\mathfrak{gl}}

\def\La{\Lambda}
\def\la{\lambda}

\def\Kc{\mathcal K}

\def\M{\mathcal M}
\def\Maps{\mathrm{Maps}}
\def\Mat{\mathrm{Mat}}
\def\Mu{\mathrm{M}}

\def\Oc{\mathcal O}

\def\one{\mathds 1}

\def\ov{\overline}

\def\Pbb{\mathbb P}
\def\Pc{\mathcal P}

\def\rf#1{(\ref{#1})}
\def\res{\operatorname*{res}}

\def\si{\sigma}
\def\sl{\mathfrak{sl}}
\def\Sc{\mathcal S}
\def\Sp{\operatorname{Sp}}
\def\Sym{\mathrm{Sym}}

\def\tr{\operatorname{tr}}
\def\Tc{\mathcal T}

\def\U{\rm U}

\def\wt{\widetilde}

\def\Y{\mathrm Y}

\def\Z{\mathbb Z}
\def\Zc{\mathcal Z}

\begin{document}

\title[Poisson Geometry of Monic Matrix Polynomials]{Poisson Geometry of Monic Matrix Polynomials}

\author{Alexander Shapiro}
\address{
Alexander Shapiro \medskip
\newline University of California, Berkeley, Department of Mathematics,
\newline Berkeley, CA, 94720, USA; \smallskip
\newline Institute of Theoretical \& Experimental Physics,
\newline 117259, Moscow, Russia; \smallskip
}

\email{shapiro@math.berkeley.edu}

\begin{abstract}
We study the Poisson geometry of the first congruence subgroup $G_1[[z^{-1}]]$ of the loop group $G[[z^{-1}]]$ endowed with the rational r-matrix Poisson structure for $G=GL_m$ and $SL_m$. We classify all the symplectic leaves on a certain ind-subvariety of $G_1[[z^{-1}]]$ in terms of Smith Normal Forms. This classification extends known descriptions of symplectic leaves on the (thin) affine Grassmannian and the space of $SL_m$-monopoles. We show that a generic leaf is covered by open charts with Poisson transition functions, the charts being birationally isomorphic to products of coadjoint $GL_m$ orbits. Finally, we discuss our results in terms of (thick) affine Grassmannians and Zastava spaces.
\end{abstract}

\maketitle

\section*{Introduction}

One of the central problems in the theory of integrable systems is the description of geometry of their phase spaces. The majority of known systems are modelled on the symplectic leaves of Poisson-Lie groups. In the present paper we provide a classification of symplectic leaves on the space $\M$ of (monic) matrix polynomials with the Yangian Poisson bracket. In order to motivate this problem and explain the obtained classification, let us recall some known results of the same nature. Probably, the most studied example of a Poisson-Lie group is a complex simple Lie group $G$ endowed with the so-called standard Poisson structure. Its symplectic leaves and corresponding discrete integrable systems were investigated in~\cite{HKKR00, KZ02}. The isomorphism classes of its symplectic leaves are in bijection with the double Bruhat cells on $G$, and thus are classified by pairs $(u,v)$ of the Weyl group elements. These results were generalized to the case of affine Kac-Moody groups with a trigonometric $r$-matrix in~\cite{Wil13}. In the Kac-Moody case, a class of discrete integrable systems~\cite{FM14, GSTV12, GSTV14, GK13} was obtained in the following way. Every double Bruhat cell is covered by a family of open charts parameterized by double reduced words representing $(u,v)$. These charts admit canonical Poisson isomorphisms and Poisson transition functions on the intersections. Compositions of the latter provide a family of discrete Poisson transformations that may be realised as compositions of cluster mutations.

Recall that affine Kac-Moody groups are certain central extensions of polynomial loop groups $G[z,z^{-1}]$. Another natural Poisson structure on $G[z,z^{-1}]$ comes from the rational $r$-matrix. In this case, the loop group $G[z,z^{-1}]$ has $G[z]$ and the kernel $G_1[z^{-1}]$ of the evaluation homomorphism $G[z^{-1}] \to G$ as Poisson subvarieties. At the same time, $G_1[z^{-1}]$ is an open set in the thin affine Grassmannian $\Gr = G[z,z^{-1}]/G[z]$. The symplectic leaves of $\Gr$ with the rational structure induced from $G[z,z^{-1}]$ were studied in~\cite{KWWY14}. Relying on the results of~\cite{LY08} for finite-dimensional groups, it was shown that the leaves of $\Gr$ are of the form $$
\Gr^\al_\beta = (G[z]z^\al \cap G_1[z^{-1}]z^\beta)G[z]/G[z],
$$
where $\al$ and $\beta$ are dominant coweights of $\g$ and $\al\ge\beta$. Therefore, as in the trigonometric case, the symplectic leaves on $\Gr$ are parameterized by combinatorial data encoded in a pair of coweights. At the same time, the symplectic leaves on $G_1[z^{-1}]$ are classified by a single coweight $\al$, or equivalently, by a pair $(\al,0)$.

Although, polynomial loop groups give rise to many examples of integrable systems, it seems that their generality is not fully satisfying. For example, $SL_n$ magnetic chains studied in~\cite{Gek95, Skl92} can not be modelled on symplectic leaves of $G[z,z^{-1}]$. The reason is that elements of the polynomial loop group $GL_m[z,z^{-1}]$ need to have polynomial inverses. In particular, this implies that in any representation their determinants are Laurent monomials. On the other hand, in order to model magnetic chains, one has to allow for loops whose determinants may be any nonzero polynomials. Thus, a natural choice for underlying Poison (ind-)variety would be the space of matrix polynomials which we consider in this paper. In order to incorporate usual techniques from the theory of Poisson-Lie groups we realise the space of matrix polynomials as an ind-subvariety in the formal loop group $GL_m((z^{-1}))$. We note that this construction admits a generalization for any complex semi-simple Lie group which we consider in the forthcoming publication.

Another motivation for studying Poisson geometry of the space $\M$ of (monic) matrix polynomials is that $\M$ serves as a quasi-classical analogue of the Yangian $\Y(\gl_m)$. Following the general idea of geometric quantization, there should exist a certain correspondence between irreducible representations of $\Y(\gl_m)$ and symplectic leaves of $\M$. From this point of view, the geometry of $\M$ has been studied in~\cite{GKL04, GKLO05}. Namely, in~\cite{GKLO05} for any dominant coweight $\al$ of $\g$ there was obtained a $\Y(\g)$-module depending on a number of parameters. This generalized the Gelfand-Zetlin type representations obtained in~\cite{GKL04}, which correspond to the case of $\g=\sl_m$ and $\al$ being the first fundamental coweight. It was also shown in~\cite{GKLO05} that the quasi-classical limits of the obtained representations are birationally isomorphic to open subsets in the space of $G$-monopoles. Finally, these representations were used to find the explicit solutions of the quantum open Toda chain and the quantum hyperbolic Sutherland model.

Over the past ten years, the results of~\cite{GKLO05} were extended in several different directions. In~\cite{KWWY14} a family of representations (again, depending on a dominant coweight $\al$ and a number of parameters) of shifted Yangians $\Y_\beta$ was constructed. In case $\beta=0$ this repeats the result of~\cite{GKLO05}. Moreover, these new representations were proven (modulo a technical conjecture) to quantize slices $\Gr^\al_\beta \subs \Gr$. It was proven in~\cite{FKRD15} that the Atiyah-Hitchin symplectic structure on the space of $G$-monopoles~\cite{AH88, FKMM99} coincides under the birational isomorphism described in~\cite{GKLO05} with the rational $r$-matrix Poisson structure on the thin affine Grassmannian. Finally, in~\cite{N09} there was established another connection between a space of monopoles and quantum integrable systems. More precisely, certain generating function on the Laumon space was shown to be the eigenfunction of the quantum trigonometric Calogero-Sutherland hamiltonian.

In the present paper we classify the symplectic leaves on the space $\M$ of monic matrix polynomials endowed with the rational $r$-matrix Poisson bracket. As a corollary we obtain a description of all symplectic leaves on a certain ind-subvariety $\G$ of the thick affine Grassmannian $G((z^{-1}))/G[z]$ for $G=SL_m$. The subvariety $\G$ consists of elements of the form $g(z)P(z)$ where $P(z)\in\Mat_m[z^{-1}]$ is a monic matrix polynomial in $z^{-1}$ and $g(z)\in\C[[z^{-1}]]$ is a formal monic power series representing an $m$-th root of $\det P(z)$. Note, that $\G$ contains the thin affine Grassmannian $\Gr$ as a Poisson ind-subvariety. The classification of symplectic leaves is given in terms of Smith Normal Forms (see section~\ref{section-SNF}), in other words, symplectic leaves are parameterized by sets of polynomials $d_1(z) \sco d_m(z)$ so that $d_{i+1}$ divides $d_i$ for $i=1 \sco m-1$ and the sum of their degrees is divisible by $m$. Let $r_i$ be the degree of $d_i(z)$ and $r_1 \spl r_m = mn$. We call $\al = (\al_1 \sco \al_m)$ with $\al_i = r_i-n$ the type of the leaf. We prove that all the leaves of the same type are Poisson birationally isomorphic. For each symplectic leaf we provide its dimension and describe its closure. Finally, we prove that a generic symplectic leaf is covered by open charts with birational Poisson transition functions. Moreover, each chart is birationally isomorphic to a product of $GL_m$ coadjoint orbits. This prepares a ground for discrete integrable systems, which we plan to study elsewhere.

Our result generalizes (for $G = SL_m$) the descriptions of symplectic leaves from~\cite{GKLO05} where leaves were classified only up to their type and of~\cite{KWWY14} where a special case of our result with $d_i(z) = z^{r_i}$ was obtained. In particular, this answers a question raised in~\cite{GKLO05} on how their description may be interpreted from the point of view of Poisson-Lie theory. Note that in both~\cite{GKLO05} and~\cite{KWWY14} the leaves were classified by coweights of the Lie algebra $\g$, thus the roots of polynomials $d_i(z)$ were lost and only the combinatorial data of their degrees survived. At the same time, it seems that the roots of polynomials $d_i(z)$ play the role of quantization parameters in~\cite{GKLO05} and~\cite{KWWY14}. We also note that under the birational isomorphism from~\cite{GKLO05}, the roots of polynomials $d_i(z)$ correspond to colored divisors in a partial compactification of the space of monopoles, also known as Zastava spaces~\cite{Bra06, BF14, FM99}.

The paper is organized as follows. In section 1 we recall basic facts from the Poisson-Lie theory and a definition of the Poisson ind-group. We also introduce ind-varieties $\M$ and $\G$, the main subjects of the present paper, and endow them with the Poisson ind-group structure. In section 2 we recall the Smith Normal Form theorem and obtain a classification of symplectic leaves on $\M$ and $\G$. Section 3 is devoted to the properties of symplectic leaves, we describe their dimensions, closures, and factorization of generic symplectic leaves. Finally, in section 4 we discuss in detail how our results are related to previous works.

\section*{Acknowledgements}

I express deep gratitude to my advisor, Nicolai Reshetikhin, for suggesting the topic of the present publication and for providing generous support and advice throughout. I am grateful to Piotr Achinger, Leonid Rybnikov, Gus Schrader, Ben Webster, and Alex Weekes for many valuable discussions and comments. I would also like to thank anonymous referees for many useful remarks that helped to improve this text. This research was supported by the NSF grant DGE-1106400 and by RFBR grant 14-01-00547.

\section{Poisson-Lie structure}

In this section we recall some basic facts on the finite-dimensional Poisson-Lie theory, Poisson ind-groups, and define the main subjects of the paper.

\subsection{Finite-dimensional theory}
\label{sub-fd-theory}

\begin{defn}
A \emph{Poisson variety} is a variety $M$ endowed with a \emph{Poisson bracket}
$$
\hc{\,,\,}\colon C^\infty(M)\otimes C^\infty(M)\to C^\infty(M)
$$
such that
\begin{itemize}
\item $C^\infty(M)$ is a Lie algebra with the bracket $\hc{\,,\,}$;
\item the Leibniz rule is satisfied, i.e. for any $\phi,\psi,\eta \in C^\infty(M)$ one has
$$
\hc{\phi\psi,\eta} = \phi\hc{\psi,\eta} + \hc{\phi,\eta}\psi
$$
\end{itemize}
\end{defn}

For any function $\phi\in C^\infty(M)$, the map
$$
\hc{\phi,-} \colon C^\infty(M)\to C^\infty(M),\quad \psi\mapsto\hc{\phi,\psi}
$$
is a derivation, thus defines a vector field $\xi_\phi \in Vect(M)$ by the formula
$$
\ha{\xi_\phi,d\psi} = \hc{\phi,\psi}.
$$
Such vector fields are called \emph{Hamiltonian}. In particular, we see that the bracket $\hc{\phi,\psi}$ depends only on $d\phi \wedge d\psi$ and there exists a \emph{Poisson bivector field} $\pi\in\Gamma(\La^2TM)$ uniquely defined by
\beq
\label{bracket}
\hc{\phi,\psi}=d\phi \otimes d\psi (\pi).
\eeq

\begin{defn}
\label{def_leaves}
A \emph{symplectic leaf} on a Poisson variety is an equivalence class of points, joined by a piecewise smooth Hamiltonian integral curve.
\end{defn}
Each symplectic leaf is an immersed Poisson subvariety bearing a symplectic structure, and any Poisson variety is a disjoint union of its symplectic leaves.

\begin{defn}
A \emph{Poisson-Lie group} is a Lie group $G$ equipped with a Poisson structure such that the group multiplication $m\colon G\times G\to G$ is a map of Poisson varieties.
\end{defn}

It is easy to show that a bivector field $\pi\in\Gamma(\La^2TG)$ defines a Poisson structure on a Lie group $G$ if and only if
\beq
\label{bivector_iff}
\pi(xy) = (d_x(\rho_y)\otimes d_x(\rho_y))\pi(x) + (d_y(\la_x)\otimes d_y(\la_x))\pi(y)
\eeq
where
$$
\la_x\colon G\to G, \quad g\mapsto xg \qquad\text{and}\qquad \rho_y\colon G\to G, \quad g\mapsto gy
$$
are respectively left and right translations on $G$. Thus, if $\pi(g)=0$ and $\Sc$ is a symplectic leaf on $G$, then $g\Sc$ is a symplectic leaf as well.

Let $G$ be a Poisson-Lie group with a bivector $\pi$ and a Lie algebra $\g=T_eG$. With the use of right translations on the tangent bundle $TG$, the bivector $\pi\in\Gamma(\La^2TG)$ defines a map $\wt\pi \colon G\to\La^2\g$ with a derivative $\de=d_e\wt\pi\colon\g\to\La^2\g$. This yields the following definition.
\begin{defn}A \emph{Lie bialgebra} is a Lie algebra $\g$ equipped with a \emph{cobracket} $\de\colon\g\to\La^2\g$ such that
\begin{itemize}
\item $\de^\ast\colon\La^2\g^\ast\to\g^\ast$ defines a Lie bracket on $\g^\ast$;
\item the \emph{cocycle condition}
$$
\de([a,b]) = (\ad_a\otimes1 + 1\otimes\ad_a)\de(b) - (\ad_b\otimes1 + 1\otimes\ad_b)\de(a)
$$
is satisfied.
\end{itemize}
\end{defn}

A classical theorem, due to Drinfeld~\cite{Dri86}, asserts that the functor $G\to Lie(G)$ between the category of connected, simply connected Poisson-Lie groups and the category of finite-dimensional Lie bialgebras is an equivalence of categories.

Let
$$
\si\colon\g\otimes\g\to\g\otimes\g, \quad a\otimes b \mapsto b\otimes a
$$
denote the permutation of tensor factors in $\g^{\otimes2}$. For any $r \in \g\otimes\g$, define the elements $r_{12}, r_{13}, r_{23} \in \g^{\otimes3}$ as follows
$$
r_{12}=r\otimes1, \quad r_{13}=(1\otimes\si)r_{12}, \quad r_{23}=1\otimes r.
$$
We call an element $r\in\g\otimes\g$ an $r$-\emph{matrix} if it satisfies the \emph{classical Yang-Baxter equation}
$$
[r_{12},r_{13}] + [r_{12},r_{23}] + [r_{13},r_{23}]=0.
$$
Let $r\in\g\otimes\g$ be an $r$-matrix whose symmetric part $r+\si(r)$ is invariant under the adjoint action of $\g$. Then, the Lie bialgebra $\g$ with a cobracket
$$
\de(a) = [1\otimes a + a\otimes1, r]
$$
is called \emph{quasitriangular}. Now, consider a Lie group $G$ whose Lie algebra $\g$ carries the structure of a quasitriangular Lie bialgebra with the $r$-matrix $r$. After trivializing the tangent bundle by right translations the bivector
$$
\pi(g) = \Ad_g(r) - r
$$
defines a Poisson-Lie structure on $G$.

For a detailed exposition of the Poisson-Lie theory we refer the reader to~\cite{CP94, ES98, KS98}.

\subsection{Poisson ind-groups}
\label{sub_Lie}

One of the motivations for this work was to study the Poisson geometry of the quasi-classical limit of Yangians in more details. This limit is naturally identified with the first congruence subgroup $G_1[[z^{-1}]]$ of the loop group $G[[z^{-1}]]$, i.e. the kernel of the evaluation map $G[[z^{-1}]] \to G$ at $z^{-1} = 0$. However, if one attempts to describe symplectic leaves on the whole $G_1[[z^{-1}]]$, they inevitably run into a problem of integrating vector fields on $\mathbb A^\infty$. One way to avoid that is to consider a ``smaller'' loop group, such as the group of analytic loops, and use analysis to study it. Here we want to stay within methods of algebra, so we consider an ind-subvariety $\M \in G_1[[z^{-1}]]$ and study its Poisson geometry using the theory developed in \cite{Wil13}.

Let us recall that an \emph{ind-variety} is a union of an increasing sequence of finite-dimensional varieties $X_n$ whose inclusions $X_n \hra X_{n+1}$ are closed embeddings. A \emph{ring of regular functions} $\C[X]$ of an ind-variety $X$ is an inverse limit
$$
\C[X] = \varprojlim \C[X_n]
$$
of the rings of regular functions on $X_n$. Given two ind-varieties $X$ and $Y$ with filtrations $X_n$ and $Y_n$ respectively, we say that $f\colon X \to Y$ is a \emph{regular map of ind-varieties} if for every $i\ge0$ there exists $n(i)\ge0$ such that $f(X_i) \subseq Y_{n(i)}$ and, moreover, $f|_{X_i}\colon X_i \to Y_{n(i)}$ is regular.

Now, \emph{Poisson ind-variety} is an ind-variety $X$ endowed with a continuous Poisson bracket $\C[X] \otimes \C[X] \to \C[X]$. An \emph{ind-group} is defined as an ind-variety $X$ with a regular group operation $X \times X \to X$. Combining the last two notions one gets the following definition.

\begin{defn}
A \emph{Poisson ind-group} $G$ is a Poisson ind-variety whose group operation $G \times G \to G$ is a regular map of Poisson ind-varieties.
\end{defn}

We refer the reader to \cite{Kum02, Wil13} for a detailed exposition of ind-groups and Poisson ind-groups.

\subsection{Poisson structure on matrix polynomials}

For the rest of this section let $G = GL_m$.

Consider the space $\Mat_m(\C[z^{-1}])$ of matrices over the ring of $\C$-valued polynomials in $z^{-1}$. Elements of this space can be treated as matrix-valued polynomials
$$
P(z) = \suml{k=0}{n} P_k z^{-k}, \qquad P_k \in \Mat_m(\C).
$$

We say that a matrix polynomial or power series in $z^{-1}$ is \emph{monic} if its constant term is the identity matrix $\one\in\Mat_m(\C)$. Let $\M\subs\Mat_m(\C[z^{-1}])$ be a set of monic matrix polynomials, and
$$
\M_n = \hc{\one + P_1 z^{-1} \spl P_n z^{-n}}, \qquad P_k\in\Mat_m(\C)
$$
be the subset of polynomials of degree at most $n$. Natural inclusions $\M_n \hra \M_{n+1}$ endow $\M$ with a structure of an ind-variety. Let $\C_1[[z^{-1}]]$ be the field of monic power series. Consider the map
$$
\C_1[[z^{-1}]] \times \M \to G_1[[z^{-1}]], \qquad \hr{g(z), P(z)} \mapsto g(z)P(z).
$$
Its image forms a subgroup in $G_1[[z^{-1}]]$ which we denote by $\wt\G$. Now, we define a subgroup $\G\subs\wt\G$ by
$$
\G = \hc{g(z)\in\wt\G \,|\, \det g(z)=1}.
$$
Elements of $\G$ are matrices $P(z)\in\M$ divided by power series $\sqrt[m]{\det P(z)}\in\C[[z^{-1}]]$. This allows us to endow $\G$ with the structure of an ind-variety in the same way is we did for $\M$. It is immediate that the group multiplication in $\G$ is a regular map of ind-varieties, thus $\G$ becomes an ind-group.

Now we endow $\M$ and $\G$ with Poisson structures. We want to define the Poisson bivector field, as in the finite-dimensional case, via an $r$-matrix which we construct with the use of Manin triples.

\begin{defn}
A \emph{Manin triple} is a triple of Lie algebras $(\a,\a_+,\a_-)$ where $\a = \a_ + \oplus \a_-$ is equipped with an invariant nondegenerate bilinear form $(\,,\,)$ such that
\begin{itemize}
\item $\a_+$ and $\a_-$ are isotropic subalgebras;
\item the form $(\,,\,)$ induces an isomorphism $\a_- \simeq \a_+^\ast$.
\end{itemize}
\end{defn}

The following proposition, due to Drinfeld, relates Manin triples and Lie bialgebras, see~\cite{ES98, KS98}.
\begin{prop}
Let $(\a,\a_+,\a_-)$ be a (possibly infinite-dimensional) Manin triple. Then the Lie bracket on $\a_-\simeq\a_+^\ast$ defines a map $\de\colon\a_+\to\La^2\a_+$ that turns $\a_+$ into a Lie bialgebra.
\end{prop}

Consider a Manin triple with $\a=\a_0((z))$, $\a_+=z^{-1}\a_0[z^{-1}]$, $\a_-=\a_0[[z]]$, $\a_0$ being a simple Lie algebra, and the bilinear form given by
\beq
\label{pairing}
\br{f(z),g(z)} = \res\limits_{z=0}\tr\br{f(z)g(z)}.
\eeq
It defines the \emph{dual Yangian} bialgebra structure \cite{ES98} on $\a_+$ with the cobracket
\beq
\label{cobracket}
\de(c z^{-n})=\suml{r=1}{n}\suml{i=1}{\dim\a_0}[x_i,c]z^{-r}\otimes x_i z^{r-n-1}
\eeq
where $(x_i)$ is an orthonormal basis of $\a_0$. Let us use two different variables $u$ and $v$ to distinguish between $\a_+=u^{-1}\a_0[u^{-1}]$ and $\a_-=\a_0[\![v]\!]$. Then $\a_+$ is a (pseudo) quasitriangular Lie bialgebra \cite{Dri86} with the $r$-matrix given by
\beq
\label{r}
r=\sums{n\ge0}\suml{i=1}{\dim\a_0} x_i u^{-n-1}\otimes x_i v^n = \frac\Omega{u-v},
\eeq
here $\Omega$ denotes the Casimir of $\a_0$, and $(u-v)^{-1}$ is expanded in the region $|u| > |v|$. Using the form~\eqref{pairing} we may consider the $r$-matrix~\eqref{r} as an element of a completed tensor product $\a_+ \widehat\otimes \a_+$, see \cite[Section 3.3]{Wil13}.

Let $f\M$, $f \in \C[[z^{-1}]]$ be the translate of the ind-variety $\M$ by an element $f$.

\begin{prop}
\label{prop_bivector}
The same formula
\beq
\label{bivector}
\pi(g) = \Ad_g(r) - r
\eeq
defines a Poisson bivector field on $f\M$ for any $f \in \C[[z^{-1}]]$.
\end{prop}

\begin{proof}
It is easy to see that expression~\eqref{bivector} satisfies property~\eqref{bivector_iff}. The $r$-matrix~\eqref{r} is invariant under conjugation by elements in $\C[[z^{-1}]]$. Therefore, $\pi(f)=0$ for any $f \in \C[[z^{-1}]]$ and we only need to show that $\Ad_g(r) - r \in \La^2T\M$ for any $g \in \M$. The semigroup $\M$ is generated by the elements
\beq
\label{elts}
\exp(E_{i,i+1}z^{-n}) \qquad\text{and}\qquad \suml{j=1}{n}E_{j,j}p_j(z)
\eeq
where $n\ge0$, $i=1 \sco m-1$, $E_{i,j}$ denote a matrix unit, and $p_j(z)$ are monic polynomials in $z^{-1}$. A straightforward check shows that $\Ad_g(r) - r \in \La^2T\M_n$ if $g$ is an element of the form~\eqref{elts} with $p_j(z)$ being a polynomial of degree less or equal to $n$.
\end{proof}

\begin{prop}
\label{prop_Jacobi}
Bivector~\eqref{bivector} defines a Poisson bracket on $\C[f\M]$ for any $f \in \C[[z^{-1}]]$.
\end{prop}

\begin{proof}
By Proposition~\ref{prop_bivector} and \cite[Proposition 3.7]{Wil13}, the bivector~\eqref{bivector} defines a continuous skew-symmetric bracket on $\C[f\M]$ satisfying the Leibniz rule. That this bracket satisfies the Jacobi identity follows from the fact that $r$ is a solution of the classical Yang-Baxter equation, see~\cite[Proposition 3.13]{Wil13}.
\end{proof}

In other words, $\wt\G$ is a family of isomorphic Poisson ind-varieties labelled by the elements of $\C_1[[z^{-1}]] / \C_1[z^{-1}]$, with isomorphisms given by translations by elements of $\C[[z^{-1}]]$.

\begin{prop}[see \cite{RSTS93}]
\label{prop_conj_inv}
If $\phi, \psi \in \C[\wt\G]$ are conjugation invariant functions, then $\hc{\phi,\psi}=0$.
\end{prop}

\begin{cor}
\begin{enumerate}
\item $\M_n$ are Poisson subvarieties of a Poisson ind-variety $\M$;
\item $\G$ is a Poisson ind-group.
\end{enumerate}
\end{cor}

\begin{proof}
Part 1 follows from Proposition~\ref{prop_Jacobi} and the proof of Proposition~\ref{prop_bivector}. In turn, proposition~\ref{prop_conj_inv} implies that the determinant of a matrix $P\in f\M$ is constant along the symplectic leaf containing $P$. Therefore, Poisson bivector~\eqref{bivector} restricts from the family $\wt\G$ to $\G$ turning the latter into a Poisson-ind group.
\end{proof}

Let us now describe the Poisson structure on $\M$ explicitly. By $t_{ij}^{(k)}\in C^{\infty}(\M)$ we denote a function that evaluates to the $(ij)$-entry of the $k$-th coefficient $P_k$ on a monic matrix polynomial $P(z)$. Consider a generating series of functions
$$
T(u) = T^{(0)} + T^{(1)}u^{-1} + T^{(2)}u^{-2} + \dots \qquad\text{where}\qquad T^{(k)} = \hr{t_{ij}^{(k)}}_{i,j=1}^{m}.
$$
It is easy to check \cite{KWWY14} that the Poisson bracket on $\M$ can be written in the Leningrad notation as follows
\beq
\label{Poisson}
\Bc{T(u) \mathop{\otimes}_{\stackrel{\displaystyle{,}}{\phantom{-}}} T(v)} = \hs{\frac{\Omega}{u-v},\;T(u) \otimes T(v)}
\eeq
where $\Omega = \suml{i,j=1}n E_{ij}\otimes E_{ji}$ is the Casimir element in $\gl_n$. Equivalently, formula~\eqref{Poisson} reads as
$$
\hc{t_{ij}(u),t_{kl}(v)} = \frac{1}{u-v}\hr{t_{kj}(u)t_{il}(v) - t_{kj}(v)t_{il}(u)},
$$
where $t_{ij}(u) = t_{ij}^{(0)} + t_{ij}^{(1)}u^{-1} + t_{ij}^{(2)}u^{-2} + \dots,$ or even more explicitly,
\beq
\label{Poisson-coord}
\hc{t_{ij}^{(r)}, t_{kl}^{(s)}} = \suml{q=\max(r,s)}{r+s-1} t_{kj}^{(r+s-q-1)} t_{il}^{(q)} - t_{kj}^{(q)}t_{il}^{(r+s-q-1)}.
\eeq
Thus, functions $t_{ij}^{(r)}$ for $r>n$ generate the defining Poisson ideal for subvarieties $\M_n$. Let $\M'_n\subs\M$ be a subvariety of matrix polynomials in $z^{-1}$ of degree exactly $n$. With the above Poisson structure $\M$ becomes a disjoint union of finite dimensional Poisson subvarieties $\M'_n$.

\begin{remark}
The Poisson bracket~\eqref{Poisson} is defined in such a way that ind-varieties $\M$ and $\G$ can be treated as classical versions of Yangians $\Y(\gl_m)$ and $\Y(\sl_m)$ respectively, see also \cite[Theorem 3.9]{KWWY14}.
\end{remark}

\section{Symplectic leaves}

In this section we classify the symplectic leaves on Poisson ind-varieties $\M$ and $\G$.

\subsection{Smith normal form}
\label{section-SNF}

\begin{theorem}
\label{th-SNF}
Let $R$ be a principal ideal domain. Then any matrix $M\in\Mat_m(R)$ can be written in the form
$$
M = A D B \qquad\text{where}\qquad A,B \in GL_m(R), \quad D=diag(d_1 \sco d_m)
$$
and $d_i$ is divisible by $d_{i+1}$ for $i=1\sco m-1$.
\end{theorem}

The theorem is standard, idea of the proof is as follows. Let $M_{i,j}$ be an entry of $M$ with the smallest norm. We first multiply $M$ by transposition matrices to place $M_{i,j}$ in the last row and the last column. We then multiply $M$ by elementary matrices to reduce norms of elements of the last row and column by subtracting multiples of $M_{n,n}$. This procedure is nothing but a Euclidian algorithm which we perform several times until we get a matrix whose only nonzero entry of the last row and the last column is $M_{n,n}$. If it happened so, that there exists an entry not divisible by $M_{n,n}$, we use an elementary row/column operation again to add this entry to the last row/column. We repeat the Euclidian algorithm until $d_n = M_{n,n}$ is the only nonzero entry in the last row and column and divides all other entries of the matrix $M$. Then we run the same process with the smaller matrix that we obtain by deleting the last row and the last column of $M$. We refer the reader to~\cite{GLR82} for a complete proof in case $R = \C[z]$.

It is clear from the proof that the elements $d_i$ are unique up to multiplication by a unit in $R$, and the product $d_{n-r+1} \dots d_n$ equals the greatest common divisor of all $r \times r$ minors of the matrix $M$. Thus $GL_m(R)$ double cosets in $\Mat_m(R)$ are precisely the Smith normal forms in $\Mat_m(R)$. The matrix $D$ is called the \emph{Smith normal form} of $M$. In case $R=\C[z]$ elements $d_i$ are called the \emph{invariant polynomials} of the matrix~$M$. Now we can formulate the main result of this section.

\subsection{Classification of symplectic leaves}

Let $\Pc_n \in \Mat_m[z]$ denote the set of monic matrix polynomials of degree $n$ (where monic means that the coefficient in front of $z^n$ is the identity matrix). Then we induce a structure of a Poisson variety on $\Pc_n$ using the isomorphism
\beq
\label{MP}
\M_n \simeq \Pc_n, \qquad P(z) \mapsto z^nP(z).
\eeq

\begin{theorem}
\label{th-leaves}
Symplectic leaves on the variety $\M_n$ are the varieties of monic matrix polynomials $P\in\M_n$ of degree~$n$ in $z^{-1}$ such that the corresponding polynomials $z^nP \in \Pc_n$ have a given Smith normal form.
\end{theorem}

The rest of this section is devoted to the proof of Theorem~\ref{th-leaves}. First, we recall an analogous result for finite-dimensional Poisson-Lie groups. Let $G$ be a (finite-dimensional) Poisson-Lie group, $G^*$ its dual, and $D$ its double (see~\cite{CP94}, \cite{ES98}, or~\cite{KS98}). Then, the following theorem holds.

\begin{theorem}
\label{th-fd-leaves}
\begin{enumerate}
\item The symplectic leaves of $G$ are the dressing orbits of $G^*$ on $G$;
\item Equivalently, the symplectic leaves of $G$ are the connected components of its intersections with the double cosets of $G^*$ in $D$.
\end{enumerate}
\end{theorem}

This theorem was proven in~\cite{LW90, STS85} for finite-dimensional $G$ and generalized to the case of Kac-Moody groups with the standard Poisson structure in~\cite{Wil13}. Although $\M$ is not a Poisson-Lie group, we can still prove a similar result. In our case, the role of $G^*$ is played by $GL_m[z]$, and varieties of polynomials with a given Smith Normal Form are nothing but $GL_m[z]$ double cosets in $\M$. According to our definition of a symplectic leaf it is enough to show that the tangent space to the leaf $\Sc$ at every point $x\in\Sc$, i.e. the span of Hamiltonian vector fields at $x$, coincides with the tangent space to the $GL_m[z]$ double cosets. The outline of the proof is as follows. We first prove Propositions~\ref{prop_dressing} and~\ref{prop-cosets} which are local versions of the statements~1 and~2 of the Theorem~\ref{th-fd-leaves} respectively. Then we show that the tangent space to a leaf $\Sc_x$ containing $x\in\M$ indeed exponentiates to an orbit of the $GL_m[z] x GL_m[z]$.

Consider the following decomposition of the vector space $\Tc=\Mat_m((z^{-1}))$
$$
\Tc = \Tc^+ \oplus \Tc^- \quad\text{where}\quad \Tc^+=\Mat_m[z] \quad\text{and}\quad \Tc^-=z^{-1}\Mat_m[[z^{-1}]].
$$
Let $\Tc^\pm_n\subs\Tc^\pm$ be the subspaces of matrix polynomials of degree at most $n$ in $z$ or $z^{-1}$ respectively. Then $\Tc^-_n$ can be identified with the tangent space $T_P \M_n$ at any point $P \in \M_n$. Denote by $A_\pm$ projections of an element $A\in\Tc$ onto $\Tc_\pm$.

\begin{prop}
\label{prop_dressing}
The tangent space $T_P\Sc$ to the symplectic leaf $\Sc\subs\M_n$ at the point $P(z)$ coincides with the space of all matrix polynomials of the form
\beq
\label{dressing}
X_A(P) = \br{PA}_+P - P\br{AP}_+
\eeq
where $A\in\Tc^+$
\end{prop}

\begin{proof}
Let $\xi_{ij}^{(r)}$ be a Hamiltonian vector field corresponding to the function $t_{ij}^{(r)}$. By formula~\eqref{Poisson-coord} we have
\beq
\label{xi-t}
\ha{\xi_{ij}^{(r)}, dt_{kl}^{(s)}} (P(z)) =
\suml{q=\max(r,s)}{\min(n,r+s-1)} \hr{t_{kj}^{(r+s-q-1)} t_{il}^{(q)} - t_{kj}^{(q)}t_{il}^{(r+s-q-1)}} (P(z))
\eeq
for any $P(z) \in \M_n$. In the above formula, pairing with the differential $dt_{kl}^{(s)}$ evaluates the $(k,l)$-entry of the coefficient in front of $z^{-s}$ of the vector $\xi_{ij}^{(r)}(P(z))$. After eliminating indices $k$ and $l$ from the formula~\eqref{xi-t}, summing over $s$, and changing $r$ to $r+1$ we get
\beq
\label{xi}
\xi_{ij}^{(r+1)} (P(z)) = \suml{s=1}{n}\hr{\suml{q=\max(s,r+1)}{\min(n,r+s)} P_{s+r-q}E_{ji}P_{q} - P_{q}E_{ji}P_{s+r-q} } z^{-s},
\eeq
where $P_0=\one$ and $E_{ji}$ is a matrix unit. Let $A=\suml{r=0}{n-1}A_{-r}z^r$ be an element of $\Tc^+_{n-1}$. Set
$$
\xi_A = \suml{r=0}{n-1}\tr\hr{ A_{-r}^t\xi^{(r+1)} }
\qquad\text{where}\qquad
\xi^{(r)} = \hr{\xi_{ij}^{(r)}}_{i,j=1}^m
$$
and $A_{-r}^t$ stands for the transpose of the matrix $A_{-r}$. The vector fields $\xi_A$ for various $A\in\Tc^+_{n-1}$ exhaust all hamiltonian vector fields on $\M_n$. Formula~\eqref{xi} implies that the vector field $\xi_A$ at point $P(z)\in\M_n$ can be written as
\beq
\label{vf}
\xi_A(P(z)) = \suml{s=1}{n}\hr{\suml{r=0}{n-1}\suml{q=\max(s,r+1)}{\min(n,r+s)} \hs{P_{s+r-q},A_{-r},P_{q}}}z^{-s},
\eeq
where
$$
\hs{P_{s+r-q},A_{-r},P_{q}} = P_{s+r-q}A_{-r}P_{q} - P_{q}A_{-r}P_{s+r-q}.
$$

Now, it is only left to show that the expressions~\eqref{dressing} and~\eqref{vf} coincide. First of all, note that it is enough to consider only $A\in\Tc^+_{n-1}$ in~\eqref{dressing}. Indeed for $A\in z^n\Mat_m[z]$ one has
$$
\br{PA}_+P - P\br{AP}_+ = PAP - PAP = 0.
$$
We will prove the rest by induction on $n$. For $n=1$ we have $P = 1 + P_1z^{-1}$, $A = A_0$, thus
$$
X_A(P) = \br{PA}_+P - P\br{AP}_+ = A_0P - PA_0 = \hs{A_0,P_1}z^{-1}.
$$
On the other hand, $n=1$ forces $s=q=1$ and $r=0$ in~\eqref{vf}, in which case
$$
\xi_A(P) = \hs{P_0,A_0,P_1}z^{-1} = X_A(P).
$$

Now, since $X_A$ is linear in $A$ it is enough to consider only $A = A_{-r}z^r$. Assume, that $\xi_A(P) = X_A(P)$ for any $k\le n$. Let
$$
P = 1 + P_1z^{-1} \spl P_n z^{-n}, \qquad \wt P = P + P_{n+1} z^{-n-1}.
$$
one has
$$
X_A(\wt P) = (PA)_+\wt P - \wt P(AP)_+ = X_A(P) + \bar X_A(\wt P)
$$
where
$$
\bar X_A(\wt P) = z^{-n-1}\br{(PA)_+P_{n+1} - P_{n+1}(AP)_+}.
$$
On the other hand,
$$
\xi_A(\wt P) = \xi_A(P) + \bar \xi_A(\wt P),
$$
where $\bar \xi_A(\wt P)$ are terms depending on $P_{n+1}$. It is easy to check that
$$
\bar \xi_A(\wt P) = \hs{P_r,A_{-r},P_{n+1}}z^{-n-1} + \suml{s=0}{r-1}\hs{P_s,A_{-r},P_{n+1}}z^{r-1-n-s} = \bar X_A(\wt P),
$$
which finishes the proof.
\end{proof}

\begin{prop}
\label{prop-cosets}
The tangent space $T_P\Sc$ coincides with the space of all matrix polynomials of the form
\beq
\label{cosets}
BP - PC \in\Tc^-_n \qquad\text{such that}\qquad B,C\in\Tc^+_{n-1}.
\eeq
\end{prop}

\begin{proof}
It is sufficient to show that the expressions~\eqref{dressing} and~\eqref{cosets} coincide. It follows from Proposition~\ref{prop_dressing} that every vector of the form~\eqref{dressing} is automatically of the form~\eqref{cosets}. Now we only need to prove the converse.

Consider an arbitrary element $B\in\Tc^+_{n-1}$. For any $P\in\M_n$, there exists a unique element $A\in\Tc^+_{n-1}$ such that
\beq
\label{prop-cosets-local}
B(z) = (PA)_+ = \suml{k=0}{n-1}\suml{i=0}{n-1-k}\hr{P_i A_{-k-i}}z^k.
\eeq
Indeed, equality~\eqref{prop-cosets-local} reads as a system of $n-1$ equations on the coefficients of $A$. We solve them inductively, starting from the coefficient $A_{1-n}$ in front of the top power of $z$. Condition $P_0 = \one$ guarantees that the solution exists and unique. Then we read~\eqref{cosets} as a system of $n-1$ equations on the coefficients of $C$. Once again, we solve them inductively, starting from the coefficient in front of the top power of $z$, and find that
$$
C(z) = (AP)_+.
$$
As before, $P_0 = \one$ insures that the solution exists and unique. This proves that every expression of the form~\eqref{cosets} admits a presentation of the form~\eqref{dressing}.
\end{proof}

{\noindent\bf Proof of Theorem~\ref{th-leaves}.\,}
It suffices to show that
\begin{enumerate}
\item[1)] for every $P(z) \in \M_n$ the double cosets $GL_m[z] P(z) GL_m[z]$ intersect $\M_n$ transversally at $P(z)$ and their intersection $\Kc = GL_m[z] P(z) GL_m[z] \cap \M_n$ is irreducible;
\item[2)] the tangent spaces $T_P\Kc$ and $T_P\Sc$ coincide.
\end{enumerate}
The first statement follows from~\cite[Theorem 1.4]{Ric92}, the argument there is given for finite-dimensional groups, but carries over to our case without issues. It is clear that $T_P\Kc \subs T_P\Sc$. Indeed, by $1)$ we have
$$
T_P(\Kc) = T_P(\M_n) \cap T_P(GL[z]P(z)GL[z])
$$
and all the vectors in the right hand side part are of the form~\eqref{cosets}. Thus, now it is only left to prove that $T_P\Kc \sups T_P\Sc$.

The Lie algebra $\gl_m[z]$ coincides with $\Tc^+$ as a vector space and is generated by elements
$$
E_{i,i+1}z^k, \qquad E_{i+1,i}z^k, \qquad \one z^k,
$$
with $i = 1 \sco n-1$ and $k\ge0$. In notations of Proposition~\ref{prop-cosets}, condition $B=\one z^k$ implies equalities $C=B$ and $BP-PC=0$. On the other hand, every generator $E_{i,j}z^k\in\gl_m[z]$ with $i \ne j$ exponentiates to the element $\one + E_{i,j}z^k \in GL_m[z]$. Therefore, any vector of the form~\eqref{cosets} with $B,C$ being some generators of the algebra $\gl_m[z]$ can be represented as a tangent vector to an orbit of $GL_m[z] \times GL_m[z]$ action on $\M$ by left and right translations. Thus, the same holds for any $B,C \in \Tc^+_{n-1}$. This finishes the proof.
\hfill $\blacksquare$

\begin{cor}
\label{cor-leaves}
Symplectic leaves on $\G$ are classified by sets of $m-1$ monic polynomials $q_1 \sco q_{m-1} \in \C[z]$ such that $q_i$ is divisible by $q_{i+1}$ for $i=1 \sco m-2$.
\end{cor}

\begin{proof}
Consider a map
$$
\M_n \to \G, \qquad P(z) \mapsto \frac{P(z)}{\sqrt[m]{\det P(z)}}.
$$
It is a Poisson projection on its image, and sends symplectic leaves to symplectic leaves. Now, the image of a symplectic leaf $\Sc\in\M$, corresponding to a Smith normal form with invariant polynomials $d_1, \sco d_m$, is characterized by polynomials $q_i = d_i/d_m$ where $i=1 \sco m-1$.
\end{proof}

\section{Properties of symplectic leaves}

In this section we describe dimensions, closures, and classes of birationally isomorphic symplectic leaves. We also show that a generic leaf on $\Pc_n$ is covered by a number of open subsets with birational Poisson transition functions, each subset being birationally isomorphic to a product of $n$ coadjoint $GL_m$ orbits.

\subsection{Closures, dimensions, isomorphism classes}

Recall the Poisson isomorphism~\eqref{MP}. For simplicity, in what follows we will be working with varieties $\Pc_n$ rather than $\M_n$.

\begin{defn}
The \emph{Smith normal form}, \emph{invariant polynomials}, and the \emph{determinant}, denoted $\det\Sc$, of a symplectic leaf $\Sc\subs\Pc_n$ are respectively the Smith normal form, invariant polynomials, and the determinant of some (thus, any) matrix~$P(z)\in\Sc$.
\end{defn}

The determinant $\det P(z)$ of any matrix $P(z)\in\Pc_n$ is a polynomial of degree $mn$. Denote by $S_k$ the symmetric group on $k$ elements and consider the map
\beq
\label{chi}
\chi_n \colon \Pc_n \longra \C^{mn}/S_{mn}
\eeq
sending a polynomial $P(z)$ to the collection of roots of $\det P(z)$. It follows from Proposition~\ref{prop_conj_inv} and the fact that $\det$ is an $\ad$-invariant function on $\Pc_n$ that the fibers $\chi_n^{-1}(x)$ of the map~\eqref{chi} are Poisson. Moreover, generic fiber is symplectic. Indeed, roots of the determinant of a generic matrix $P(z)$ are pairwise distinct. This together with conditions that $d_{i+1}$ divides $d_i$ for $i = 1 \sco m-1$ and $\det P(z) = d_1(z) \sd d_m(z)$ forces invariant polynomials to satisfy $d_2 \sco d_m \equiv 1$ and $d_1 = \det P(z)$. Thus, for generic $x\in\C^{mn}/S_{mn}$ the fiber $\chi_n^{-1}(x)$ is a symplectic leaf.

Recall the definition of the thin affine Grassmannian
$$
\Gr = G[z,z^{-1}]/G[z].
$$
We can endow $\Gr$ with a Poisson structure in exact same way as it was done for $\M$. This structure will coincide with the one discussed in~\cite{KWWY14}. Let $G_1[z^{-1}]$ be the first congruence subgroup of the loop group $G[z^{-1}]$. For $G = GL_m$ we get
$$
G_1[z^{-1}] = \M \cap G[z,z^{-1}]
$$
where the intersection is taken inside $G((z^{-1}))$. On the other hand,
$$
G_1[z^{-1}] \simeq G_1[z^{-1}]G[z]/G[z] \subset \Gr
$$
is a Poisson subvariety of $\Gr$ (see~\cite{KWWY14}). Define $G_1[z^{-1}]_n \eqdef \M_n \cap G[z,z^{-1}]$.

\begin{prop}
We have an isomorphism of Poisson varieties
\beq
\label{Gr-thin}
G_1[z^{-1}]_n \simeq \chi_n^{-1}(0).
\eeq
\end{prop}

\begin{proof}
Indeed, $G_1[z^{-1}]$ can be described as a subgroup of $\M$ consisting of elements $P(z)$ invertible in $G[z,z^{-1}]$. This is equivalent to the condition that $\det P(z)$ is invertible as a Laurent polynomial, hence a Laurent monomial. Since $\det P(z) \in 1 + z^{-1}\C[z^{-1}]$ for any $P(z) \in \M_n$, we get
$$
G_1[z^{-1}]_n = \hc{P(z) \in \M_n \,|\, \det P(z) = 1}.
$$
On the other hand,
$$
\chi_n^{-1}(0) = \hc{P(z) \in \Pc_n \,|\, \det P(z) = z^{mn}}.
$$
Hence, the isomorphism~\eqref{Gr-thin} is the restriction of the isomorphism~\eqref{MP}.
\end{proof}

Let $d_1 \sco d_m$ be the invariant polynomials of a symplectic leaf $\Sc\subs\Pc_n$ with degrees $r_1 \sco r_m$ respectively. Clearly, $r_1 \spl r_m = mn$. Set
$$
\al = (\al_1 \sco \al_m) \qquad\text{where}\qquad \al_i = r_i - n.
$$

\begin{defn}
We call $\al$ the \emph{type} of the symplectic leaf $\Sc$.
\end{defn}

Evidently, types of symplectic leaves on $\M$ may be identified with dominant coweights of the Lie algebra $\sl_m$.

\begin{prop}
\label{prop-isom}
Two leaves of the same type are Poisson birationally isomorphic.
\end{prop}

\begin{proof}
Given an $m \times m$ matrix $A$ denote by $A_{i_1 \sco i_r}^{j_1 \sco j_r}$ the minor sitting in the intersection of rows $i_1 \sco i_r$ and columns $j_1 \sco j_r$. The following functions
\beq
\label{Drinfeld}
\bar a_i(A) = A_{i+1 \sco m}^{i+1 \sco m}, \qquad \bar b_i(A) = A_{i+1 \sco m}^{i, i+2 \sco m} \qquad\text{for}\qquad i=1\sco m-1
\eeq
serve as classical analogue of Drinfeld's new coordinates~\cite{Dri88} on $\Y(\gl_m)$. Let $\Sc$ be a symplectic leaf with invariant polynomials $d_1 \sco d_m$. For any $P(z)\in\Sc$, the function $a_i(z) = \bar a_i(P(z))$ is a monic polynomial in $z$ of degree $n(m-i)$, and $b_i(z) = \bar b_i(P(z))$ is a polynomial in $z$ of degree less than $n(m-i)$. Define rational functions
\beq
\label{monopole}
e_i(z) = \frac{b_i(z)}{a_i(z)}, \qquad i=1 \sco m-1.
\eeq
For a generic $P(z) \in \Sc$ the greatest common divisor of $a_i(z)$ and $b_i(z)$ is equal to the product $d_{m-i+1} \sd d_m$, hence the function $e_i(z)$ has $k_i$ poles, where
$$
k_i = n(m-i) - r_{m-i+1} - \dots - r_m.
$$
Let $\Sc^\circ$ be the open subset in $\Sc$ where functions $e_i(z)$ have $k_i$ simple poles. For $i = 1 \sco m-1$ and $s = 1 \sco k_i$ let $x_{i,s}$ denote the poles of $e_i(z)$, and set $y_{i,s}$ to be the residues of $e_i(z)$ at $x_{i,s}$. Based on results of~\cite{FKMM99, FM99} it was shown in~\cite{GKLO05} that parameters $(x_{i,s}, y_{i,s})$ define \'etale coordinates on $\Sc^\circ$.

Now, note that the \'etale coordinates $(x_{i,s}, y_{i,s})$ and the Poisson brackets between them (see~\cite{GKLO05, FKRD15}) depend only on the type of the leaf $\Sc$ but not on the set of its invariant polynomials. Hence, for a pair of symplectic leaves $\Sc_1$ and $\Sc_2$ of the same type we get that their open subsets $\Sc_1^\circ$ and $\Sc_2^\circ$ are Poisson isomorphic.
\end{proof}

\begin{cor}
Leaves of type $\al$ are of dimension
$$
2\ha{\al,\rho} = \suml{i=1}{m}(m+1-2i)\al_i = \suml{i=1}{m}(m+1-2i)r_i
$$
where $\rho$ is the half sum of all positive roots of $\sl_m$.
\end{cor}

\begin{proof}
Although, it is easy to calculate the dimension of a symplectic leaf directly, let us take a different approach. By Proposition~\ref{prop-isom}, the dimension of a symplectic leaf depends only on its type. Note, that for any symplectic leaf $\Sc \subs \Pc_n$ of type $\al$ there exists a symplectic leaf $\Sc' \subs \chi_n^{-1}(0)$ of the same type. Now, the statement follows from the isomorphism~\eqref{Gr-thin}, and the dimensions of symplectic leaves in the thin affine Grassmannian, see e.g.~\cite{MV00}.
\end{proof}

\begin{prop}
The closure $\ov\Sc \subset \Pc_n$ of a symplectic leaf $\Sc$ of type $\alpha$ is an affine variety of dimension $2\ha{\al,\rho}$. Let $\Sc$ and $\Sc'$ be a pair of leaves with invariant polynomials $(d_1 \sco d_m)$ and $(d'_1 \sco d'_m)$ respectively. Then $\Sc' \subset \ov\Sc$ if and only if $\det\Sc = \det\Sc'$ and the product $d_1 \dots d_k$ is divisible by $d'_1 \dots d'_k$ for all $1 \le k \le m$.
\end{prop}

\begin{proof}
The proof repeats the one of~\cite[Proposition 2.6]{BL94}.
\end{proof}

\subsection{Factorization on generic leaves}

We start with a few more definitions.

\begin{defn}
Consider a matrix $P(z)\in\Pc_n$. The $mn$ roots of the polynomial $\det P(z)$ are called the \emph{eigenvalues} of $P(z)$. For any eigenvalue $\la$ the matrix $P(\la)$ has a nonzero kernel, elements of this kernel are called the \emph{eigenvectors} of $P(z)$ corresponding to the eigenvalue $\la$.
\end{defn}

\begin{defn}
The \emph{spectrum} $\Sp(\Sc)$ of a symplectic leaf $\Sc$ is the collection of eigenvalues of some (thus any) polynomial $P(z)\in\Sc$. Symplectic leaf is said to be \emph{generic} if its eigenvalues are pairwise distinct.
\end{defn}

Clearly, generic leaves are just fibers $\chi_n^{-1}(x)$ of $\Pc_n$ over a generic point $x\in\C^{mn}/S_{mn}$.

\begin{defn}
A matrix polynomial $z-A$, $A\in\Mat_m(\C)$, is said to be a \emph{right divisor} of $P(z)\in\Pc_n$ if $P(z) = Q(z)(z-A)$ for some polynomial $Q(z)\in\Pc_{n-1}$.
\end{defn}

The following lemma appears in~\cite{Bor04}.

\begin{lemma}
\label{lem-div}
Consider a matrix polynomial $P(z) \in \Pc_n$ with distinct eigenvalues. Let $\la_1\sco\la_m$ be some eigenvalues of $P(z)$ and $v_1 \sco v_m$ be the corresponding eigenvectors. Assume that $v_1 \sco v_m$ are linearly independent, and consider a matrix $A\in\Mat_m(\C)$ defined by $Av_i = \la_iv_i$ for $i=1 \sco m$. Then $z-A$ is a right divisor of $P(z)$. Moreover, if the roots of $\det P(z)$ are pairwise distinct then $A$ is uniquely defined by the conditions that $z-A$ is a right divisor of $P(z)$ and $\Sp(A) = \hc{\la_1 \sco \la_m}$.
\end{lemma}

\begin{cor}
\label{cor-factor}
Let $\Sc\subs\Pc_n$ be a generic leaf. Then any matrix polynomial $P(z)\in\Sc$ admits a decomposition into linear factors.
\end{cor}

\begin{proof}
It suffices to show that $P(z)$ has a linear right divisor. Assume that it does not. Then by Lemma~\ref{lem-div} all the eigenvectors of $P(z)$ lie in some $k$-dimensional subspace of $\C^m$, $k<m$. Let $Q(z)$ be the restriction of $P(z)$ onto this subspace. Then $Q(z)$ is a monic matrix polynomial of degree at most $n$ and coefficients in $\Mat_k(\C)$. Its determinant $\det(Q(z))$ is a polynomial of degree at most $dn$. On the other hand, any of $mn$ distinct roots of $\det(P(z))$ is also a root of $\det(Q(z))$, and we arrive at a contradiction.
\end{proof}

Let $\gl_m^*$ be the dual space to the Lie algebra $\gl_m$. We endow $\gl_m^*$ with the Kirillov-Kostant-Souriau Poisson bracket and a product of $n$ copies of $\gl_m^*$ with the product Poisson structure. The following lemma is classical, and goes back at least to~\cite[Chapter IV, \S3]{FT87}.

\begin{lemma}
\label{lem_Poisson}
Map
\beq
\label{product}
\underbrace{\gl_m^* \times \dots \times \gl_m^*}_n \to \Pc_n, \qquad (A_1 \sco A_n) \mapsto (z-A_1)\dots(z-A_n).
\eeq
is Poisson
\end{lemma}

We summarize results of this section in the following theorem.

\begin{theorem}
\label{th-factor}
Let $\Sc$ be a generic symplectic leaf on $\Pc_n$. Then
\begin{enumerate}
\item[1)] For any ordered partition
\begin{gather*}
\La = \hr{\La_1 \sco \La_n}, \qquad \La_i = \hc{\la_{i,1} \sco \la_{i,m}}, \\
 \la_{i,j}\in\Sp(\Sc), \qquad i=1 \sco n, \quad j=1 \sco m
\end{gather*}
of the spectrum of $\Sc$ there exists an open subset $\Sc_\La\subs\Sc$ and a birational Poisson isomorphism
$$
\phi_\La\colon \Sc_\La \to \Oc_{\La_1} \st \Oc_{\La_n}
$$
Here $\Oc_{\La_i}$ is the adjoint orbit on $\gl_m$ with the fixed spectrum $\La_i$, and $\phi_\La$ is inverse to the product map given by the formula~\eqref{product}. The subset $\Sc_\La\subs\Sc$ is described by conditions that vectors $v_1 \sco v_m$ from Lemma~\ref{lem-div} are linearly independent for each $\La_i$, $i=1 \sco n$. The symplectic leaf $\Sc$ is covered by the union of open subsets $\Sc_\La$.
\item[2)] For any pair $\La$ and $\Mu$ of ordered partitions of $\Sp(\Sc)$ there exists a birational Poisson map
$$
\tau_{\La\Mu} \colon \phi_\Mu(\Sc_\La \cap \Sc_{\Mu}) \to \phi_\La(\Sc_\La \cap \Sc_{\Mu}), \qquad \tau_{\La\Mu} = \phi_{\La}\circ(\phi_\Mu)^{-1}.
$$
\end{enumerate}
\end{theorem}

\begin{proof}
Part 1 is equivalent to the statement that there exists an open set $\Sc_\La$ in a generic leaf $\Sc$, such that any polynomial $P(z)\in\Sc_\La$ admits decomposition into linear factors
$$
P(z) = (z-A_1) \dots (z-A_n), \qquad \Sp(A_i) = \La_i
$$
and such decomposition is unique. Both existence and uniqueness follow from Lemma~\ref{lem-div}. Indeed, $P(z)$ admits $z-A_n$ as a right divisor if the eigenvectors with eigenvalues in $\La_n$ are linearly independent, in which case $A_n$ is uniquely defined by $\La_n$. Then proceed by induction on $n$. In turn, Corollary~\ref{cor-factor} ensures that $\Sc$ is covered by the union of $\Sc_\La$. Part~2 is obvious.
\end{proof}

Note that for any map $\tau_{\La\Mu}$ can be written as a composition of maps exchanging a pair of eigenvalues of adjacent matrices $A_k$ and $A_{k+1}$. Propositions~\ref{prop_ortog} and~\ref{prop_trans} below provide a more explicit description of the map $\tau_{\La\Mu}$ by dealing with the case $n=2$. Assume that a matrix polynomial $P(z)$ of degree 2 admits decompositions $P(z)=(z-A)(z-B)$ and $P(z)=(z-\wt A)(z-\wt B)$, where
\begin{align}
&\Sp(A) = \hc{\la_1,\la_2\sco\la_m}, &&\Sp(B) = \hc{\mu_1,\mu_2\sco\mu_m},
\\ \label{tilde_spec}
&\Sp(\wt A) = \hc{\mu_1,\la_2\sco\la_m}, &&\Sp(\wt B) = \hc{\la_1,\mu_2\sco\mu_m},
\end{align}
and all $\la_i$, $\mu_j$, $i,j=1\sco m$ are distinct. In the rest of this section we set $\la = \la_1$ and $\mu=\mu_1$ for brevity.

\begin{prop}
\label{prop_ortog}
Consider vectors $u,v\in\C^m$ such that $Av = \la v$ and $u^tB = \mu u^t$. Then their scalar product $(u,v)$ is nonzero.
\end{prop}
\begin{proof}
Let $w_i$ be the eigenvectors of $B$ with eigenvalues $\mu_i$ respectively. Note that $u$ is not an eigenvector of $B$ but of $B^t$. For $i=2 \sco m$ vector $u$ is orthogonal to $w_i$. Indeed,
$$
\mu(u,w_i) = \mu u^tw_i = u^tBw_i = \mu_i u^tw_i = \mu_i(u,w_i)
$$
which yields $(u,w_i)=0$. Assume that $(u,v)=0$. Then $v$ is a linear combination of vectors $w_i$, $i=2 \sco m$.

Consider a matrix polynomial $P(z) = (z-A)(z-B)$. It has an eigenvector $w = (\la-B)^{-1}v$ with eigenvalue $\la$. Note that
$$
u^tw = \frac{1}{\la-\mu}u^tv = 0.
$$
Therefore, $w$ also lies in a linear combination of vectors $u_i$, where $i=2 \sco m$. Now, if $P(z) = (z-\wt A)(z-\wt B)$, the matrix $\wt B$ has eigenvectors $w$ and~$u_i$, $i=2 \sco m$. Then the restriction of $\wt B$ onto the $(m-1)$-dimensional subspace generated by $u_i$, $i=2 \sco m$, has $m$ distinct eigenvalues. Thus, we arrive at a contradiction and~$(u,v)\ne0$.
\end{proof}

\begin{prop}
\label{prop_trans}
One has
\beq
\label{tildeAB}
\wt A = A + (\mu-\la)T, \qquad \wt B = B + (\la-\mu)T
\eeq
where $T$ is a projector onto $v$ along $u^t$.
\end{prop}
\begin{proof}
The projector $T$ can be written as
\beq
\label{T}
T = \frac{vu^t}{(u,v)}
\eeq
and is well defined due to Proposition~\ref{prop_ortog}. Consider vectors $v_i,u_i\in\C^m$ such that $Av_i = \la_iv_i$, $u_i^tB = \mu_iu_i^t$ for $i = 2 \sco m$. Note that unlike in Proposition~\ref{prop_ortog}, vectors $u_i$ are not eigenvectors of $B$ but of $B^t$, as well as vector $u$. A straightforward check shows that vectors
$$
\wt v_i = v_i + \frac{\mu-\la}{\la_i-\mu}\frac{(u,v_i)}{(u,v)}
\qquad\text{and}\qquad
\wt u_i = u_i + \frac{\la-\mu}{\mu_i-\la}\frac{(u_i,v)}{(u,v)}
$$
satisfy $\wt A\wt v_i = \la_i\wt v_i$, $\wt u_i^t\wt B = \mu_i\wt u_i^t$ for $i = 2 \sco m$. One also has $\wt Av = \mu v$ and $u^t\wt B = \la u^t$. Thus, spectra of $\wt A$ and $\wt B$ are as in~\eqref{tilde_spec}.

To prove that $(z-A)(z-B) = (z-\wt A)(z-\wt B)$ one needs to check that $A+B = \wt A + \wt B$ and $AB = \wt A = \wt B$. The first equality is obvious from~\eqref{tildeAB}. Using that $T^2=T$ since $T$ is a projector, and that $AT = \la T$ and $TB = \mu T$, which follows from formula~\eqref{T}, one has
$$
\wt A \wt B = AB + (\la-\mu)\la T + (\mu-\la)\mu T - (\la-\mu)^2T^2 = AB.
$$
This finishes the proof.
\end{proof}

\section{Discussion}

Let us remind some of the previous results on the Poisson geometry of affine Grassmannians and moduli spaces of monopoles. In~\cite{GKL04} a series of infinite-dimensional \emph{Gelfand-Zetlin type} $\U(\gl_m)$-modules is constructed. The underlying vector space for these modules is the space of meromorphic functions in $m(m-1)/2$ variables on which elements of the algebra $\U(\gl_m)$ act by difference operators. The corresponding actions are classified by generic $\gl_m$ coadjoint orbits, or equivalently, by a set of $m$ distinct complex numbers. The authors also consider pull-backs of these modules under the evaluation homomorphism $\Y(\gl_m) \to \U(\gl_m)$. This construction is generalized in~\cite{GKLO05} to the case of any complex semi-simple Lie algebra. Namely, it is shown that for any dominant coweight $\al$ of the algebra $\g$ and a certain collection of parameters $c_i$ there exists a $\Y(\g)$-module, which in case $\g=\sl_m$ and $\al$ being the first fundamental coweight coincides with (the restriction to $\Y(\sl_m)$ of) one of the modules from~\cite{GKL04}. Next, relying on the fact that the constructed $\Y(\g)$-modules are defined by the action of ratios of Drinfeld's new coordinates~\cite{Dri88}, the authors of~\cite{GKLO05} suggest that functions~\rf{monopole} define a set of coordinates on an open set in every symplectic leaf, and that their degrees, considered as algebraic maps $\Pbb^1 \to \Pbb^1$, classify symplectic leaves on $\M$. We remark that this is not totally correct, the degrees only determine the type of a leaf.

Let us remind that the \emph{moduli space of $G$-monopoles} is a space of based holomorphic maps from $\Pbb^1$ to the flag variety $X = G/B_-$ sending $\infty$ to $B_- \in X$. The map $f\colon\Pbb^1 \to X$ is said to be of degree $\al$ where $\al \in H_2(X,\Z)$ is a positive coweight of $\g$, if for the fundamental class $[\Pbb^1] \in H_2(\Pbb^1,\Z)$ we have $f_*[\Pbb^1]=\al$. We denote the space of maps of degree $\al$ by $\Maps^\al(\Pbb^1,X)$. Note, that the maps
\beq
\label{open}
\Pbb^1 \to \hr{\Pbb^1 \times \dots \times \Pbb^1}, \qquad \infty \mapsto (0 \sco 0)
\eeq
of multidegree $\al$ form an open subset in $\Maps^\al(\Pbb^1,X)$. This observation allows the authors of~\cite{GKLO05} to identify open subsets of symplectic leaves on $\M$ with open subsets in the moduli space of $G$ monopoles with the use of functions~\rf{monopole}. The space $\Maps^\al(\Pbb^1,X)$ admits a partial compactification, also known as \emph{Zastava space}. Set-theoretically, Zastava space is described as follows
\beq
\label{zastava}
\Zc^\al = \sums{\beta\le\al} \Maps^\beta(\Pbb^1,X) \times \Sym^{\al-\beta}(\C),
\eeq
where $\Sym^{\al-\beta}$ is the space of \emph{colored divisors} of the form $\sum\ga_ic_i$ where $c_i\in\C$ and $\ga_i$ are positive coweights of $\g$ satisfying $\sum\ga_i = \al-\beta$. Note, that the space of based maps $\Pbb^1 \to \Pbb^1$ of degree $r$ coincides with the space of rational functions $P(z)/Q(z)$ where $Q$ is a monic polynomial of degree $r$, $P$ is a polynomial of degree less than $r$, and they do not have common roots. In the same vein, the space~\rf{open} may be thought of as a collection of such rational functions. Then, an open subset of the Zastava space~\rf{zastava} is again a collection of such rational functions with the condition on the absence of common roots dropped. From this viewpoint, the colored divisors encode the same information as the Smith normal forms of symplectic leaves described in the article. We refer the reader to~\cite{Bra06, BF14, FM99} and references therein for details on Zastava spaces.

Now, let $\al$ and $\beta$ be a pair of dominant coweights of $\g$. Denote by
$$
\Gr^\al = G[t]t^\al \qquad\text{and}\qquad \Gr_\beta = G_1[t^{-1}]t^\beta
$$
the corresponding orbits in the thin affine Grassmannian $\Gr = G[t,t^{-1}]/G[t]$. Let us consider subvarieties in $\Gr$ of the form
$$
\Gr^\al_\beta = \Gr^\al \cap \Gr_\beta \qquad\text{and}\qquad \Gr^{\ov\al}_\beta = \ov{\Gr^\al} \cap \Gr_\beta,
$$
where $\Gr^{\ov\al} = \bigsqcup_{\ga\le\al}\Gr^\ga$ is the closure of $\Gr^\al$. It is known that $\Gr^{\ov\al}_\beta$ is nonempty if and only if $\al\ge\beta$. The results of~\cite{GKLO05} were generalised in~\cite{KWWY14}. First, it was shown~\cite[Theorem 2.5]{KWWY14} that subvarieties $\Gr^\al_\beta$ are symplectic leaves in $\Gr$. Second, given a slice $\Gr^{\ov\al}_\beta$ and a collection of parameters $\bar{\bf c}$ there was constructed a family of representations of the shifted Yangian $\Y_\beta(\g)$, coinciding with representations from~\cite{GKLO05} for $\beta=0$. It was shown (modulo some technical conjecture) that the image $\Y^\la_\beta(\bar{\bf c})$ of $\Y_\beta(\g)$ quantizes the algebra $\Oc(\Gr^{\ov\al}_\beta)$. It is known~\cite{BF14} that for any slice $\Gr^{\ov\al}_\beta$ there exists a birational map
\beq
\label{s}
s^\al_\beta \colon \Gr^{\ov\al}_\beta \to \Zc^{\al-\beta}
\eeq
given in local coordinates~\rf{monopole} by
$$
\Gr^{\ov\al}_\beta \ni g \mapsto \hr{e_1(g) \sco e_r(g)} \in (\Pbb^1 \times \dots \times \Pbb^1) \subs \Zc^{\al-\beta}.
$$
Recently~\cite{FKRD15} the map $s^\al_\beta$ was shown to be Poisson.

Let us position the results of this paper in the above background. For $\g = \sl_m$ Corollary~\rf{cor-leaves} sharpens the description of symplectic leaves given in~\cite{GKLO05}, it does not only classify symplectic leaves up to their type but provides the full set of their invariants. On the other hand, it generalises the description of leaves on an open set $G_1[t^{-1}] \subs \Gr$ from~\cite{KWWY14}, the latter correspond to the leaves of $\G$ with all roots of invariant polynomials being 0. In particular, this explains the connection between the descriptions of~\cite{GKLO05} and~\cite{KWWY14}, indeed there exists a unique leaf of each type with spectrum consisting only of zeros. Moreover, it seems that the parameters of quantization used in both~\cite{GKLO05} and~\cite{KWWY14} encode the spectra of the leaves. Finally, it follows from~\cite{FKRD15} and the present article that with $\beta=0$ the map~\rf{s} extends to a Poisson map $\bar s^\al \colon \ov{\G^\al} \to \Zc^\al$, sending the roots of greatest common divisors of coordinates $a_i$ and $b_i$ from~\rf{Drinfeld} to the colored divisors in $\Zc^\al$. Moreover, restriction of $\bar s^\al$ to any symplectic leaf of type $\al$ is birational. Here $\ov{\G^\al}$ denotes the Poisson submanifold of $\G$ consisting of all the leaves of types $\beta\le\al$. This provides a way to realise colored divisors as Smith Normal Forms of symplectic leaves on $\G$ and extends known dictionary between affine Grassmannians and $G$-monopoles.

We conclude this paper with some unanswered questions. First, we think that the results of this paper should be generalised to semi-simple Lie algebras of arbitrary type and to Poisson structure given by trigonometric $r$-matrix. Second, we are curious to find the precise relation between quantization of the symplectic leaves obtained here, and the families of Yangian modules constructed in~\cite{GKLO05, KWWY14}. Finally, maps $\tau_{\La\Mu}$ from the Theorem~\ref{th-factor} suggest that tensor products of $\Y(\gl_m)$-modules constructed in~\cite{GKL04} admit intertwining operators labelled by permutations of the joint spectra of coadjoint orbits. In our opinion, it would be interesting to find precise formulas for these operators. We return to these questions in the forthcoming publication.

\end{document}